\documentclass[conference]{IEEEtran}

\usepackage[cmex10]{amsmath}
\usepackage{amssymb}
\usepackage{url}
\usepackage{epsfig}
\usepackage{bm}
\usepackage{psfrag}
\usepackage{graphicx}
\usepackage{amsmath}
\usepackage{algorithm}
\usepackage{algorithmic}
\usepackage{subfigure}
\usepackage{cases}
\usepackage{enumerate}

\newtheorem{lemma}{Lemma}

\newtheorem{proof}{Proof}

\begin{document}

\title{A New Distributed Localization Method for Sensor Networks}

\author{\IEEEauthorblockN{Yingfei Diao}
\IEEEauthorblockA{School of Control Science\\ and Engineering\\
Shandong University\\
Jinan, Shandong 250061\\ China\\
Email: yfdiao@gmail.com}
\and
\IEEEauthorblockN{Zhiyun Lin}
\IEEEauthorblockA{College of Electrical \\Engineering\\
Zhejiang University\\
Hangzhou, Zhejiang 310007\\China\\
Email: linz@zju.edu.cn}
\and
\IEEEauthorblockN{Minyue Fu}
\IEEEauthorblockA{School of Electrical Engineering \\
and Computer Science\\
University of Newcastle\\
Callaghan, NSW 2308 \\Australia\\
Email: minyue.fu@newcastle.edu.au}
\and
\IEEEauthorblockN{Huanshui Zhang}
\IEEEauthorblockA{School of Control Science\\ and Engineering\\
Shandong University\\
Jinan, Shandong 250061\\China\\
Email: hszhang@sdu.edu.cn}}

\maketitle

\begin{abstract}
This paper studies the problem of determining the sensor locations in a large sensor network using relative distance (range) measurements only. Our work follows from a seminal paper by Khan et al. \cite{diloc} where a distributed algorithm, known as DILOC, for sensor localization is given using the barycentric coordinate. A main limitation of the DILOC algorithm is that all sensor nodes must be inside the convex hull of the anchor nodes. In this paper, we consider a general sensor network without the convex hull assumption, which incurs challenges in determining the sign pattern of the barycentric coordinate. A criterion is developed to address this issue based on available distance measurements. Also, a new distributed algorithm is proposed to guarantee the asymptotic localization of all localizable sensor nodes.
\end{abstract}

\IEEEpeerreviewmaketitle

\section{Introduction}
Location based service is the fundamental issue in the area of sensor networks and it requires to solve the localization problem. The localization problem consists of two parts, namely, acquiring measurements and transforming them to coordinate information. In this paper, we consider the situation of using relative distance measurements only and focus on the localization procedure.

Existing work on localization can be divided into two classes \cite{cricket}: sequential methods and concurrent methods. A sequential method begins with a set of anchor nodes and computes the locations of other nodes one by one or group by group. A prominent example is the so-called  \emph{trilateration} method. Its advantage is that it is easy to implement, but it requires each location-unknown node to have three connections (edges) with location-known nodes, which is a sufficient but not necessary condition for localizability.

A concurrent method starts with some initial estimate for the coordinate of every sensor node. Each node then updates its coordinate in a distributed or cooperative manner using the relative distance measurements with its neighbors and the estimates of the neighbors' coordinates. The iterative process terminates when the estimates converge, hopefully to the true coordinates.

A novel concurrent method called distributed iterative localization (DILOC) was given by \emph{Khan et. al.} ~in \cite{diloc} based on the barycentric coordinate representation for sensor localization. The unique feature of this method is that the sensor locations can be expressed as a linear system, which can be solved iteratively due to the desirable eigenvalue distribution of the linear system. Their method relies on two key assumptions on the network topology. First, all location-unknown nodes must be inside the \emph{convex hull} of the anchor nodes. Second, every node, other than the anchor nodes, must be inside a triangle of three neighboring nodes. The main reasons of requiring these two assumptions are to avoid the difficulties in determining the sign patterns of the barycentric coordinate and to ensure the desired eigenvalue distribution of the linear system so that the iterative algorithm asymptotically converges. These two assumptions are kind of restrictive, especially for a large sensor network when the sensing range for each node is limited and when the neighbors for each node cannot be arbitrarily arranged.

In this paper, we follow the work of \cite{diloc} by generalizing their method and eliminating the two key assumptions as mentioned above. The main idea is to employ a general form of barycentric coordinate representation which allows the coordinate of each node to be expressed as a linear function of the coordinates of any three neighbors. A criterion is developed, according to which the linear function can be determined using the relative distance measurements alone. The implication of this result is that all the sensor locations can be expressed as a linear system, just like in the standard DILOC case. However, the new linear system may not have a desired eigenvalue distribution like the standard DILOC algorithm to work. We then provide a new distributed iterative algorithm for localization. This is done by applying a diagonal pre-conditioner to the linear system. A simulation result is also provided to validate the effectiveness of our proposed algorithm.


{\bf Notations:} $\mathbb{C}$ denotes the set of complex number. $\mathbb{E}$ denotes the Euclidean space. $\mathbf{1}_n$ represents the $n$-dimensional vector of ones and
$I_n$ denotes the identity matrix of order $n$. The bold font of letter indicates vector and capital letter indicates matrix. $\Delta_{ijk}$ denotes a triangle formed by node $i$, $j$ and $k$.

\hfill

\section{Preliminaries and Problem Formulation}\label{sec2}
\setcounter{equation}{0} 

\subsection{Barycentric coordinates}

The barycentric coordinate, which was firstly introduced by August Ferdinand M\"obius in 1827 \cite{cox}, is a geometric notion characterizing the relative position of one node with respect to its several neighbor nodes. For one node, say $l$ with its Euclidean coordinate $p_l$, and its three neighbor nodes, say $i$, $j$ and $k$ with their Euclidean coordinates $p_i$, $p_j$ and $p_k$ in the plane, node $l$'s barycentric coordinate with respect to $i$, $j$ and $k$ is $\{a_{li},a_{lj},a_{lk}\}$ satisfying \begin{equation}\label{bcd}
p_l=a_{li}p_i+a_{lj}p_j+a_{lk}p_k.
\end{equation}
Especially, when $a_{li}+a_{lj}+a_{lk}=1$, the barycentric coordinate is called the \emph{areal} coordinate because it can be expressed as a ratio of signed areas between specified triangles. As shown in Fig.~\ref{fig:tri}, the barycentric coordinate $\{a_{li}, a_{lj}, a_{lk} \}$ is given by
\begin{equation}\label{def:a}
  \left\{
   \begin{array}{l}
a_{li}=\frac{S_{\Delta ljk}}{S_{\Delta ijk}}\\
a_{lj}=\frac{S_{\Delta lki}}{S_{\Delta ijk}}\\
a_{lk}=\frac{S_{\Delta lij}}{S_{\Delta ijk}}
   \end{array}
   \right.
\end{equation}
 where $S_{\Delta ljk}$, $S_{\Delta lki}$, $S_{\Delta lij}$ and $S_{\Delta ijk}$ are the signed areas of the corresponding triangles $\Delta ljk$, $\Delta lki$, $\Delta lij$ and $\Delta ijk$. These areas can be calculated with pairwise internode distance measurements through {\em Cayley-Menger~determinant} \cite{cayley}. For instance,
\begin{eqnarray}\label{cme}
S_{\Delta ljk}^{2}=-\frac{1}{16}
\left|
\begin{array}{cccc}
              0 &  1 & 1 &1\\
              1 &  0 & d_{lj}^2 & d_{lk}^2\\
              1 & d_{jl}^2  & 0 & d_{jk}^2\\
              1  & d_{kl}^2 & d_{kj}^2 & 0

              \end{array}
                             \right|
                            \end{eqnarray}
where $d_{lj}$, $d_{lk}$ and $d_{jk}$ are the distance measurements among node $l$, $j$ and $k$, respectively. The sign of $S_{\Delta ljk}$ is positive if node $l$ is on the left-hand side when one moves from node $j$ to $k$, and negative otherwise.

\begin{figure}[h]
\begin{center}
      \psfrag{i}{$i$}
      \psfrag{j}{$j$}
      \psfrag{k}{$k$}
      \psfrag{l}[1c]{$l$}
      \psfrag{dli}{$d_{li}$}
      \psfrag{dlj}[1c]{$d_{lj}$}
      \psfrag{dlk}{$d_{lk}$}
      \includegraphics[width=0.35\textwidth] {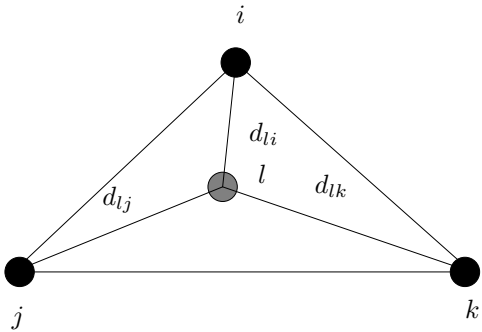}
      \caption{An illustrative example for the barycentric coordinate. }
      \label{fig:tri}
\end{center}
\end{figure}

To avoid the case that $S_{\Delta ijk}=0$, we need an assumption on the configuration of node $l$'s three neighbors as below.
\begin{enumerate}[{\bf A}0:]
\item For each node in the network, its three neighbors are not collinear.
\end{enumerate}

Note that, the computation of the coefficients in terms of \eqref{def:a}, i.e., $a_{li}$, $a_{lj}$ and $a_{lk}$, depends on the signed value of $S_{\Delta ljk}$, $S_{\Delta lki}$ and $S_{\Delta lij}$. If we only know the pairwise distance measurements, we can compute the square values of these areas and thus the absolute values of $a_{li}$, $a_{lj}$ and $a_{lk}$. However, we cannot determine the signs of these areas and thus are not able to have the barycentric coordinate.

\subsection{Problem formulation}

The common used trilateration scheme for computing the coordinate of node $l$ is to solve a group of equations like
\begin{equation}\label{def:d}
  \left\{
   \begin{array}{l}
d_{li}=\|p_l-p_i\|\\
d_{lj}=\|p_l-p_j\|\\
d_{lk}=\|p_l-p_k\|.
   \end{array}
   \right.
\end{equation}
Here, $p_u$, $u\in\{i, j, k, l\}$, is the Euclidean coordinate of node $u$ and $d_{uv}$, $u, v\in\{i, j, k, l\}$, is the distance measurement between node $u$ and $v$. These equations can be solved in a sequential way if each node has at least three distance measurements to other nodes that know their coordinates. We call those sensor nodes, who initially know their coordinates, the \emph{anchor nodes}. In contrast, we call the nodes, who do not know their coordinates initially, the normal sensor node or just \emph{sensor nodes} for short. In this paper, a network under consideration is assumed to contain at least three anchor nodes, which is a necessity for uniquely localizing the network.

Instead of solving these nonlinear equations in a sequential way, \emph{Khan et. al.} provide an iterative algorithm \cite{diloc}, named $DILOC$, to compute the coordinates of a network $\mathcal{G}$ based on the barycentric coordinate presentation.

Given a network $\mathcal{G}$, containing $n$ nodes, its Euclidean coordinates can be written in a form like

\begin{equation}\label{eq_cac}
\mathbf{p}=A\mathbf{p},
\end{equation}
where $\mathbf{p}\in\mathbb{C}^n$ is the aggregated Euclidean coordinate of $\mathcal{G}$ and the nonzero inputs in the $i$-th row, $i\in\{1,2,\cdots,n\}$ are the barycentric coordinate of node $i$. For example, if node $l$ has three neighbors $i$, $j$ and $k$, then the $l$-th row of $A$ has nonzero entries in the positions corresponding to the $i$-th, $j$-th and $k$-th columns. The other entries in the $l$-th row are all zeros.

If we consider the first three rows of $\mathbf{p}$ to represent the positions of the three anchor nodes in $\mathcal{G}$, then in terms of the natural partition of anchor nodes and normal sensor nodes, we can partition $A$ and $\mathbf{p}$ as
\begin{equation}\label{ebp}
A=
\begin{bmatrix}
I_3 &0\\
B   & C
\end{bmatrix},~~
\mathbf{p}=
\begin{bmatrix}
\mathbf{p}_a\\
\mathbf{p}_s
\end{bmatrix}
\end{equation}
where $\mathbf{p}_a$ and $\mathbf{p}_s$ correspond to the aggregate positions of anchor nodes and other normal sensor nodes, respectively. The nonzero entries of the $l$-th row can be recognized as the weights of its three neighbors. So, all diagonal inputs of $C$ are zeros. The matrix $A$ is often treated as the {\em adjacency matrix} of the network. Thus, the representation of $\mathbf{p}_s$ can be written as
\begin{equation}\label{i_diloc}
\mathbf{p}_s=C\mathbf{p}_s+B\mathbf{p}_a
\end{equation}
and
\begin{equation}\label{i_diloc2}
(I-C)\mathbf{p}_s=B\mathbf{p}_a.
\end{equation}

According to eq. (\ref{i_diloc}), we could obtain an iterative algorithm to solve $\mathbf{p}_s$ under certain conditions,
\begin{equation}\label{edi2}
\mathbf{z}_s(t+1)=C\mathbf{z}_s(t)+B\mathbf{p}_a(t)
\end{equation}
where $\mathbf{z}_s \in\mathbb{C}^{n-3}$ represents the estimate of the coordinate $\mathbf{p}_s$ of the normal sensor nodes in $\mathcal{G}$.

In \cite{diloc}, a sensor network to be localized is assumed to satisfy the following two assumptions.
\begin{enumerate}[{\bf A}1:]
\item All sensor nodes lie inside a convex hull formed by the anchor nodes.
\item Each sensor node $l$ lies inside a convex hull formed by its three neighbors. \end{enumerate}

In the 2D case, a convex hull is actually a triangle formed by three nodes. Thus, assumption {\bf A}2 leads to two constraints of the barycentric coordinate, i.e.,
\begin{eqnarray}
a_{li}+a_{lj}+a_{lk}=1,\label{cons1}\\
0<a_{li},a_{lj},a_{lk}<1.\label{cons2}
\end{eqnarray}
In \cite{diloc}, {\em Khan et. al.} proved that the spectral radius of $C$ is less than 1 when (\ref{cons1}) and (\ref{cons2}) hold. Thus, the estimate of the coordinate in (\ref{edi2}) can converge to the true value in the Euclidean coordinate system.

In this paper, we will relax the assumptions {\bf A}1 and {\bf A}2 and address a distributed algorithm to compute the locations of sensor nodes no matter whether they lie inside a convex hull or not. After relaxing these two assumptions, two problems need to be addressed. First, how to determine the signs of the barycentric coordinate when one node lies outside the convex hull of its neighbors. Second, when the convex hull assumption is dropped, the matrix $C$ in system (\ref{edi2}) might not be Schur. Then, how to provide a convergent iterative algorithm to compute the coordinate.

\section{Sign pattern determination for the barycentric coordinate}\label{sec3}

For $u, v\in\{i, j, k, l\}$, we use $\sigma_{uv} \in\{1, -1\}$  to indicate the sign of $a_{uv}$. Suppose node $l$ is localizable. It is known that no matter a node $l$ lies inside the convex hull of its three neighbors or not, the barycentric coordinate $\{a_{li}, a_{lj}, a_{lk}\}$ obtained from eq.~(\ref{def:a}) must satisfy $a_{li}+a_{lj}+a_{lk}=1$. Thus, given $|a_{li}|, |a_{lj}|, |a_{lk}|$, which can be calculated from \eqref{cme}, the problem of determining the sign pattern of the barycentric coordinate is equivalent to solve the following equation
\begin{equation}\label{eqsbc}
\sigma_{li}|a_{li}|+\sigma_{lj}|a_{lj}|+\sigma_{lk}|a_{lk}|=1
\end{equation}
where $\sigma_{li}, \sigma_{lj}$ and $\sigma_{lk}$ take values either $1$ or $-1$.

In the following we will discuss whether~\eqref{eqsbc} has a unique solution.  Moreover, if $\sigma_{li}, \sigma_{lj}$ and $\sigma_{lk}$ can not be uniquely solved from~\eqref{eqsbc}, we then explore other range based conditions to determine the sign pattern.

The first case that~\eqref{eqsbc} does not have a unique solution is that one of $|a_{li}|, |a_{lj}|, |a_{lk}|$ equals to zero. That is, a node $l$ lies on the line aligned with one of three edges of  the triangle formed by its three neighbors, according to (\ref{def:a}). Without loss of generality, say $a_{li}=0$. For this case, $\sigma_{li}$ can be either $1$ or $-1$. But the other two signs $\sigma_{lj}$ and $\sigma_{lk}$ can be determined according to the following criterion.
\begin{equation}\label{eqzero}
\{\sigma_{li}, \sigma_{lj}, \sigma_{lk}\}=
\left\{
   \begin{array}{ll}
\{\sigma_{li}, 1, 1\} & \text{if } |a_{lj}|, |a_{lk}|<1,\\
\{\sigma_{li}, 1, -1\} & \text{if }|a_{lj}|>1, |a_{lj}|>|a_{lk}|,\\
\{\sigma_{li}, -1, 1\} & \text{if }|a_{lk}|>1, |a_{lk}|>|a_{lj}|.
   \end{array}
   \right.
\end{equation}

If node $l$ does not lie on the boundary lines, there are totally 7 possible sign patterns as the pattern $\{-1, -1, -1\}$ is not possible due to~\eqref{eqsbc}. According to~(\ref{def:a}) and the definition of the signed areas, the seven possible sign patterns of $\{\sigma_{li}, \sigma_{lj}, \sigma_{lk}\}$  are shown in Fig.~\ref{fig:bc7}.
\begin{figure}[h]
\begin{center}
      \psfrag{b}{$b$}
      \psfrag{j}[lb]{$j$}
      \psfrag{k}[1b]{$k$}
      \psfrag{i}{$i$}
      \psfrag{1}[1c]{$\{1,1,1\}$}
      \psfrag{2}[1c]{$\{-1,1,1\}$}
      \psfrag{3}[1c]{$\{-1,-1,1\}$}
      \psfrag{4}[1c]{$\{1,-1,1\}$}
      \psfrag{5}[1c]{$\{1,-1,-1\}$}
      \psfrag{6}[1c]{$\{1,1,-1\}$}
      \psfrag{7}[cb]{$\{-1,1,-1\}$}
      \includegraphics[width=0.4\textwidth] {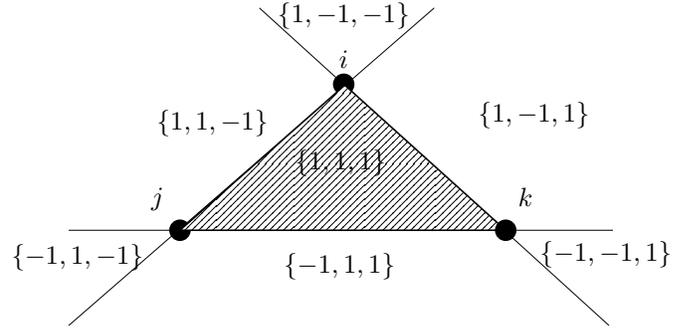}
      \caption{Seven possible sign patterns for $\{\sigma_{li}, \sigma_{lj}, \sigma_{lk}\}$.}
      \label{fig:bc7}
\end{center}
\end{figure}
In the following lemma, we characterize the second case when the sign pattern can not be uniquely solved from~\eqref{eqsbc}.

\begin{lemma}
Given $|a_{li}| \neq 0$, $|a_{lj}| \neq 0$, and $|a_{lk}| \neq 0$, the solution of~(\ref{eqsbc}) does not result in a unique sign pattern $\{\sigma_{li}, \sigma_{lj}, \sigma_{lk}\}$  if and only if one of them, saying $a_{li}$, satisfies $|a_{li}|=1$, and $|a_{lj}|=|a_{lk}|$.
\end{lemma}
\begin{proof}
(Sufficiency) If $|a_{li}|=1$ and $|a_{lj}|=|a_{lk}|$, it can be inferred from~(\ref{eqsbc}) that $\{\sigma_{li}, \sigma_{lj}, \sigma_{lk}\}=\{1, 1,-1\}$ or $\{\sigma_{li}, \sigma_{lj}, \sigma_{lk}\}=\{1, -1,1\}$. That is, (\ref{eqsbc}) does not result in a unique sign pattern.

(Necessity) Suppose there are two sign patterns both satisfying (\ref{eqsbc}). That is, it holds that
\begin{equation}\label{eqlemp}
\begin{bmatrix}
|a_{li}| & |a_{lj}| &|a_{lk}|
\end{bmatrix}
\begin{bmatrix}
\mathbf{v}_1\\
\mathbf{v}_2\\
\mathbf{v}_3
\end{bmatrix}
=
\begin{bmatrix}
1 & 1
\end{bmatrix}
\end{equation}
where $\mathbf{v}_1, \mathbf{v}_2, \mathbf{v}_3 \in \{[1~1], [-1~1], [-1~ -1], [1~ -1]\}$. This means a positive combination of
$\mathbf{v}_1, \mathbf{v}_2, \mathbf{v}_3$ equals to $[1~1]$ (see Fig.~\ref{fig:ax}). Consequently, there must be $[1~1]$ for one of  $\mathbf{v}_1, \mathbf{v}_2$, and $ \mathbf{v}_3$. Without loss of generality, we assume $\mathbf{v}_1=[1~1]$.
Next, we consider different choice of $\mathbf{v}_2$. If $\mathbf{v}_2$ equals to $[1,1]$ or $[-1, -1]$, we will have $\mathbf{v}_3$ equal to $[-1~-1]$ or $[1~1]$. In this way, two solutions $\{\sigma_{li}, \sigma_{lj}, \sigma_{lk}\}$ are actually identical.
If $\mathbf{v}_2$ equals to $[-1~1]$ or $[1~-1]$, $\mathbf{v}_3$ must equal to $[1~-1]$ or $[-1~1]$. Then according to~(\ref{eqsbc}), we know that $|a_{li}|=1$ and $|a_{lj}|=|a_{lk}|$. \end{proof}

\begin{figure}[h]
\begin{center}
      \psfrag{1}[lc]{$1$}
      \psfrag{-1}[cb]{$-1$}
      \includegraphics[width=0.25\textwidth] {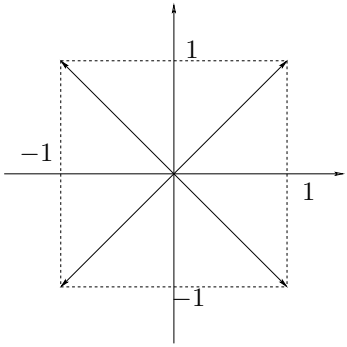}
      \caption{An illustration for the necessity proof. }
      \label{fig:ax}
\end{center}
\end{figure}

Next, we present a result on how to determine the sign pattern using the range based information when it can not be uniquely solved from~(\ref{eqsbc}).

\begin{lemma}
Given $|a_{li}|=1$ and $|a_{lj}|=|a_{lk}| \neq 0$, suppose $\angle ijk$ is an acute angle\footnote{If $\angle ijk$ is not  acute, then $\angle ikj$ must be acute and the conditions in the lemma can be modified accordingly.}.
\begin{enumerate}
\item $\sigma_{li}=-1$ if and only if \[d_{jl}=d_{ik}, \ d_{kl}=d_{ij}, \text{ and } d_{il}^2=2d_{ij}^2+2d_{ik}^2-d_{jk}^2.\]
\item If $\sigma_{li}=-1$, then $\{\sigma_{li}, \sigma_{lj}, \sigma_{lk}\}=\{-1, 1, 1\}$.
\item If $\sigma_{li}=1$ and $d_{jl}^2<d_{ij}^2+d_{il}^2$, then $\{\sigma_{li}, \sigma_{lj}, \sigma_{lk}\}=\{1, 1, -1\}$.
\item If $\sigma_{li}=1$ and $d_{jl}^2>d_{ij}^2+d_{il}^2$, then $\{\sigma_{li}, \sigma_{lj}, \sigma_{lk}\}=\{1, -1, 1\}$.
\end{enumerate}
\end{lemma}

\begin{proof}
1) (Necessity) If $\sigma_{li}=-1$, then we have $a_{li}=-1$. Moreover, since $|a_{lj}|=|a_{lk}|$, it follows from~(\ref{eqsbc}) that $a_{lj}=a_{lk}=1$. Thus, $S_{\Delta ljk}=-S_{\Delta ijk}$, $S_{\Delta lij}=S_{\Delta ijk}$ and $S_{\Delta lki}=S_{\Delta ijk}$. Comparing the sign pattern with the ones described in Fig.~\ref{fig:bc7}, we know that the only option of $l$ is at the location of $l^{'''}$ in Fig.~\ref{ediloc_sign}, which forms a parallelogram together with nodes $i$, $j$ and $k$. Hence, we can obtain directly that $d_{jl}=d_{ik}$, $d_{kl}=d_{ij}$. Furthermore, according to the parallelogram law, we have
\begin{equation}\label{eq:lll}
d_{il}^2=2d_{ij}^2+2d_{ik}^2-d_{jk}^2.
\end{equation}

(Sufficiency) If $d_{jl}=d_{ik}$ and $d_{kl}=d_{ij}$, we can draw two circles centered at $j$ and $k$ with radius $d_{ik}$ and $d_{ij}$, respectively. These two circles will have two intersection points. One of the two intersection points is $l^{'''}$ and we denote the other by $l^*$. From the necessity proof, it is known that when node $l$ is at the location of $ l^{'''}$, it satisfies~\eqref{eq:lll}. On the other hand, we will show that when node $l$ is at the location of $ l^{*}$, it does not satisfy~\eqref{eq:lll}. (To see this, it remains to show that $d_{il^{'''}}\neq d_{il^*}$. Suppose by contradiction that $d_{il^{'''}} = d_{il^*}$. Then, recalling the fact $d_{jl^{'''}}=d_{jl^*}$ and $d_{kl^{'''}}=d_{kl^*}$, we have that nodes $i$, $j$ and $k$  are on the perpendicular bisector of the line segment $l^{'''}l^*$ and so they are colinear, a contradiction to assumption A0.) Therefore, it can be concluded that due to the condition $d_{il}^2=2d_{ij}^2+2d_{ik}^2-d_{jk}^2$, node $l$ must lie at the location of $l^{'''}$.  Thus, according to the sign patterns described in Fig.~\ref{fig:bc7}, we can obtain that $\sigma_{li}=-1$.

\begin{figure}[h]
\begin{center}
      \psfrag{i}{$i$}
      \psfrag{j}{$j$}
      \psfrag{k}{$k$}
      \psfrag{k1}{$k^{'}$}
      \psfrag{l0}{$l^{'''}$}
      \psfrag{l1}{$l^{'}$}
      \psfrag{l2}{$l^{''}$}
      \includegraphics[height=0.3\textwidth] {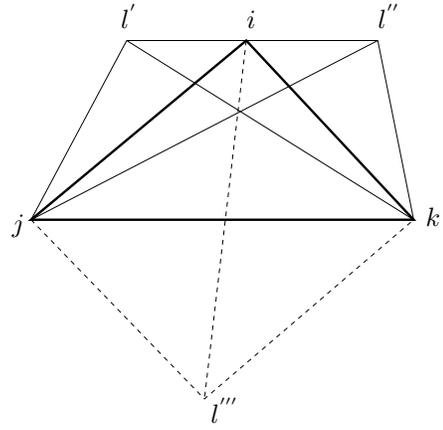}
      \caption{An example of $\Delta ijk$ and node $l$.}
      \label{ediloc_sign}
\end{center}
\end{figure}

2) If $\sigma_{li}=-1$, then $\{\sigma_{li}, \sigma_{lj}, \sigma_{lk}\}=\{-1, 1, 1\}$, which is shown in the necessity proof of 1).

3) If $\sigma_{li}=1$, then we have $a_{li}=1$. For this case,  $S_{\Delta ljk}=S_{\Delta ijk}$ according to~(\ref{def:a}). Therefore, node $l$ must be on the line that is parallel to the edge $jk$ and crosses node $i$. On this line, there are two nodes, saying node $l^{'}$ and $l^{''}$ as shown in Fig.~\ref{ediloc_sign}, whose distances to node $i$ are equal to $d_{il}$.

For the triangle $\Delta{ijl^{'}}$, according to the cosine law, it holds that
\[
d_{jl^{'}}^2=d_{ij}^2+d_{il^{'}}^2-2d_{ij}d_{il^{'}}\cos\angle jil^{'}.
\]
Since $l^{'}l^{''}$ is parallel with $jk$, we have $\angle jil^{'}=\angle ijk$. Then we know $d_{jl^{'}}^2<d_{ij}^2+d_{il^{'}}^2$ because  $\angle ijk$ is an acute angle. Similarly, for the triangle $\Delta{ijl^{''}}$, we could obtain that $d_{jl^{''}}^2>d_{ij}^2+d_{il^{''}}^2$.

Therefore, if $d_{jl}^2<d_{ij}^2+d_{il}^2$, then node $l$ must be at the locatio of $l'$. Thus, according to the sign patterns described in Fig.~\ref{fig:bc7}, we obtain that $\{\sigma_{li}, \sigma_{lj}, \sigma_{lk}\}=\{1, 1, -1\}$.

4) Following the argument in 3), we know that if $d_{jl}^2>d_{ij}^2+d_{il}^2$, then node $l$ must lie at the location of $l^{''}$. Then again from the sign patterns described in Fig.~\ref{fig:bc7}, we obtain that $\{\sigma_{li}, \sigma_{lj}, \sigma_{lk}\}=\{1, -1, 1\}$. \end{proof}

Finally, we summarize the above results to provide an algorithm of determining the sign pattern based on the range measurement information. The pseudo code is given in Algorithm~\ref{alg:sign}.

\algsetup{indent=2em}
\newcommand{\factorial}{\ensuremath{\mbox{\sc Cooperative Expression }}}
\begin{algorithm}[h!]
\caption{Determining the sign pattern of node $l$'s barycentric coordinate. }\label{alg:sign}
{\bf Input:} $|a_{li}|$, $|a_{lj}|$, $|a_{lk}|$, $d_{li}$, $d_{lj}$, $d_{lk}$, $d_{ij}$, $d_{ik}$, $d_{jk}$. \\
{\bf Output:} $\{\sigma_{li}, \sigma_{lj}, \sigma_{lk}\}$.
\begin{algorithmic}
\STATE{Solve eq.~(\ref{eqsbc})}
\IF{the solution is unique}
\STATE{{\bf Return} $\sigma_{li}, \sigma_{lj}, \sigma_{lk}$}
\ELSIF{one of $|a_{li}|$, $|a_{lj}|$, $|a_{lk}|$ equals to 0}
\STATE{Determine $\{\sigma_{li}, \sigma_{lj}, \sigma_{lk}\}$ according to~(\ref{eqzero})}
\ELSIF{$d_{jl}=d_{ik}$, $d_{kl}=d_{ij}$ and $d_{il}^2=2d_{ij}^2+2d_{ik}^2-d_{jk}^2$ }
\STATE{$\{\sigma_{li}, \sigma_{lj}, \sigma_{lk}\}=\{-1,1,1\}$}
\ELSIF{$d_{jl}^2<d_{ij}^2+d_{il}^2$} 
\STATE{$\{\sigma_{li}, \sigma_{lj}, \sigma_{lk}\}=\{1, 1, -1\}$}
\ELSIF{$d_{jl}^2>d_{ij}^2+d_{il}^2$}
\STATE{$\{\sigma_{li}, \sigma_{lj}, \sigma_{lk}\}=\{1, -1, 1\}$}
\ENDIF
\end{algorithmic}\label{alg:sign}
\end{algorithm}

\section{Distributed localization algorithm}\label{sec4}

In this section, we develop a convergent iterative algorithm for localization, that is, to solve $\mathbf{p}_s$ from~\eqref{i_diloc}. Notice that for a general barycentric coordinate, the spectral radius of $C$ may be larger than 1. It means that the iterative form presented in~(\ref{edi2}) may not converge. In the sequel, we will modify it to have a new convergent iterative algorithm to compute the coordinates of localizable sensor nodes. 

Pre-multiplying a diagonal matrix $K$ on both sides of (\ref{i_diloc}), we can obtain
\[
K\mathbf{p}_s=KC\mathbf{p}_s+KB\mathbf{p}_a.
\]
Adding $\mathbf{p}_s$ on both sides of the above equation and packing the terms, we have
\[
\mathbf{p}_s=(I-K(I-C))\mathbf{p}_s+KB\mathbf{p}_a.
\]
Recall that $\mathbf{p}_a$ is constant. So we could consider the following iterative algorithm 
\begin{equation}\label{sys}
\mathbf{z}_s(t+1)=(I-K(I-C))\mathbf{z}_s(t)+KB\mathbf{p}_a,
\end{equation}
where $\mathbf{z}_s$ is the estimate of the coordinates of the normal sensor nodes. Here, the key is to find an appropriate diagonal matrix $K$ so that  $I-K(I-C)$ is Schur. 

The above iterative form of localization can be implemented in a distributed way, that is, 
\begin{equation}\label{isys}
z_i(t+1)=z_i(t)-k_i\left(z_i(t)-\underset{j\in\mathcal{N}_i}{\Sigma}a_{ij}z_j(t)\right)
\end{equation}
where $z_i$ is the estimate of sensor node $i$'s coordinate, $k_i$ is the $i$th diagonal entry of $K$, and $a_{ij}$, $j \in \mathcal{N}_i$, is the barycentric coordinate of node $i$ with respect to its three neighbors. 

However, finding a diagonal pre-conditioner $K$ in a completely distributed way to make $I-K(I-C)$ Schur is a challenging task. Next we introduce a cluster-based approach for the design of $K$, which is partially distributed and does not require to collect all the information of the whole network. 
For a sensor network $\mathcal{G}$, if the anchor nodes are not far away from each other, then suppose $\mathcal{G}$ can be partitioned into a set of clusters $\mathcal{G}_0$, $\mathcal{G}_1$, $\dots$, $\mathcal{G}_m$ such that $\mathcal{G}_0$ is the cluster of anchor nodes, and for any cluster $\mathcal{G}_s$, $s\in\{1,\cdots,m\}$, the neighbors used to define the barycentric coordinate of any node in $\mathcal{G}_s$ belong to $\mathcal{G}_0 \cup \cdots \cup \mathcal{G}_{s}$. This is the case when the set of clusters is sequentially localizable. It is common in practice as the neighboring topology is usually dependent on  the configuration in the Euclidean space and limited communication ranges. How to detect such a set of sequential clusters is of independent interest. Readers may refer to \cite{yang} and \cite{wang}.
 
For each cluster $\mathcal{G}_s$, $s=1, \dots, m$, we propose a scheme to design $K_s$ (the corresponding block of $K$ for the cluster $\mathcal{G}_s$). The idea is inspired by the work of \cite{cdc} and \cite{au2012} using the continuity property of eigenvalues with respect to the values of $K_s$. Denote by $L_s$ the square block submatrix in $I-C$ corresponding to the cluster $\mathcal{G}_s$ (taking the rows and columns of $I-C$ indexed by the nodes in $\mathcal{G}_s$). Suppose the total number of nodes in $\mathcal{G}_s$ is $n_s$. We denote $K_s=\text{diag}(k_{s_1},k_{s_2},\dots,k_{s_{n_s}})$ and denote by $L_s^j$, $j=1, \dots, n_s$, the $j$th principal submatrix of $L_s$.  The procedure of designing $K_s$ for each cluster $\mathcal{G}_s$ is summarized in Algorithm~\ref{alg:find2}.

\begin{algorithm}[htb]
\caption{Finding the diagonal pre-conditioner $K_s$ for cluster $\mathcal{G}_s$.}
\label{alg:find2}
\begin{algorithmic}[1]
\STATE{\bf{Input:} $L_s$}
\STATE{\bf{Output:} $K_s=\textnormal{diag}(k_{s_1},k_{s_2},\dots,k_{s_{n_s}})$}
\FOR{$j=1:n_s$}
\STATE{Find $k'_{s_j}$ such that the eigenvalues of $\text{diag}\{k'_{s_1},\dots,k'_{s_j} \}L_s^j$ lie in the open right-half-plane. }
\ENDFOR
\STATE{Find sufficiently small $\varepsilon>0$ to make the eigenvalues of $\varepsilon \text{diag}\{k'_{s_1},\dots,k'_{s_{n_s}} \}L_s$ inside the unit circle centered at $(1,0)$.}
\RETURN{$k_{s_j}=\varepsilon k'_{s_j}$, $j =1, \dots, n_s$.}
\end{algorithmic}
\end{algorithm}

 Next, we discuss why Algorithm~\ref{alg:find2} can succeed in finding the diagonal pre-conditioner $K_s$. 
According to the loop in the algorithm, it is certain that an appropriate $k'_{s_1}$ can be found first so that $k'_{s_1}L_s^1$ is in the open right-half-plane. Thus, the eigenvalues of
\[
\begin{bmatrix}
k'_{s_1} & 0 \\
0 &0
\end{bmatrix}
L_s^1
\]
are $k'_{s_1}L_s^1$ and $0$. By the continuity property of eigenvalues, we can then find  $k'_{s_2}$ in the neighborhood of the origin such that the eigenvalues of
\[
\begin{bmatrix}
k'_{s_1} & 0\\
0 & k'_{s_2} \\
\end{bmatrix}
L_s^2
\]
both lie in the open right-half-plane. Repeating the argument leads to the finding of a set of $k'_{s_j}$, $j = 1, \dots, n_s$, such that the eigenvalues of $\text{diag}\{k'_{s_1},\dots,k'_{s_{n_s}} \}L_s$ all lie in the open right-half-plane, for which, a small $\varepsilon>0$ can then be chosen to shrink the eigenvalues of $\varepsilon \text{diag}\{k'_{s_1},\dots,k'_{s_{n_s}} \}L_s$ inside the unit disk centered at $(1,0)$.

To implement the algorithm of finding a diagonal pre-conditioner for the cluster $\mathcal{G}_s$, a randomly selected node in the cluster acts as a cluster head and collects the barycentric coordinates of nodes in the same cluster, i.e., $L_s$. It computes an appropriate $K_s$ according to Algorithm~\ref{alg:find2} and then sends $k_{s_j}$ to the individual nodes in the cluster. If the cluster size is medium, the required communication cost is acceptable.  

After obtaining the diagonal pre-conditioners for all clusters, the localization algorithm~\eqref{sys} is fully distributed, requiring only the exchange of the estimate from the neighbors. Moreover, it can be known that the iterative algorithm~\eqref{sys} is globally asymptotically convergent due to its linear form, while most existing localization work (e.g., \cite{subgra}) based the sub-gradient method only ensures local convergence.

\section{Simulation}\label{sec5}

In this section, a sensor network with 12 nodes is considered. As shown in Fig.~\ref{ori}, three anchor nodes are connected by black lines and other nine sensor nodes are marked by red stars. The blue lines with arrows represent the neighboring topology in localization. In this example, both assumptions A1 and A2 are not satisfied. 

 The 12 nodes are grouped into four clusters, i.e.,  $\mathcal{G}_0 =\{1, 2, 3\}$, $\mathcal{G}_1=\{4, 5, 6\}$, $\mathcal{G}_2=\{7, 8, 9\}$, and $\mathcal{G}_3=\{10, 11, 12\}$. The barycentric coordinate of each node is calculated according to Algorithm~\ref{alg:sign} based on the range measurement information. The diagonal pre-conditioner $K_s$, $s=1, 2, 3$, is obtained utilizing Algorithm~\ref{alg:find2}. The coordinate of each node is then iteratively calculated in terms of~\eqref{isys} in a distributed way. The trajectories of the estimates are shown in Fig.~\ref{trace}, from which it is seen that the estimates asymptotically converge to the true coordinates from an arbitrary initial guess.  

The residual error defined as $\frac{||\mathbf{z}_s(t)-\mathbf{p}_s||}{||\mathbf{z}_s(0)-\mathbf{p}_s||}$ is plotted in Fig.~\ref{converge} with respect to the iteration steps. Though the localization algorithm is executed in parallel, the convergence process takes three stages, as being observed in  Fig.~\ref{converge}, because the localization of a cluster depends on the localization of the cluster closer to the anchor nodes.  

\begin{figure}[h]
\centering
      \subfigure[Original network topology.]{
      \includegraphics[width=0.395\textwidth] {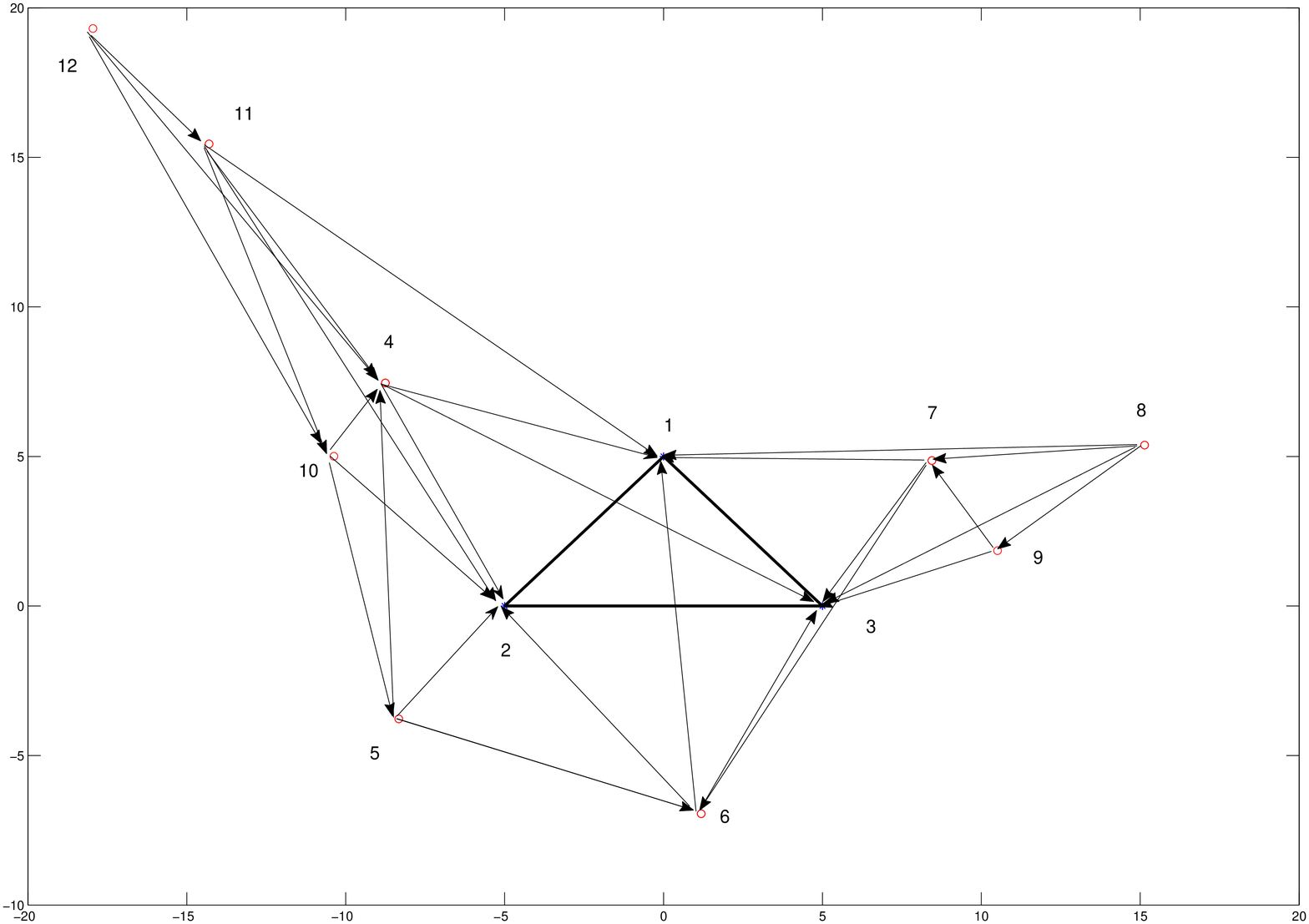}
      \label{ori}}
      \subfigure[Trajectories of the coordinate estimates.]{
      \includegraphics[width=0.395\textwidth] {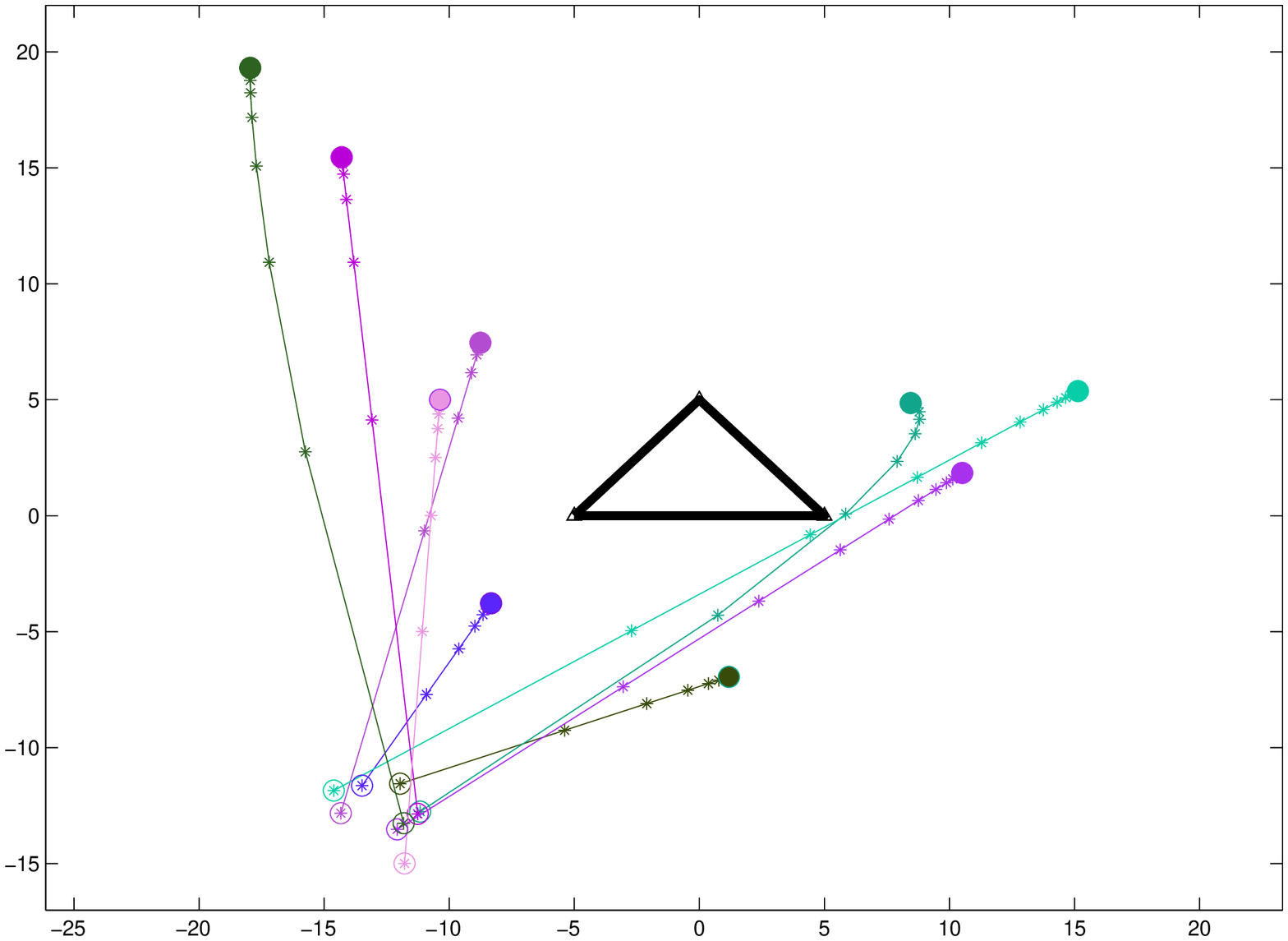}
      \label{trace}}
      \subfigure[$\frac{||\mathbf{z}_s(t)-\mathbf{p}_s||}{||\mathbf{z}_s(0)-\mathbf{p}_s||}$ w.r.t. $t$]
      {
      \includegraphics[width=0.395\textwidth] {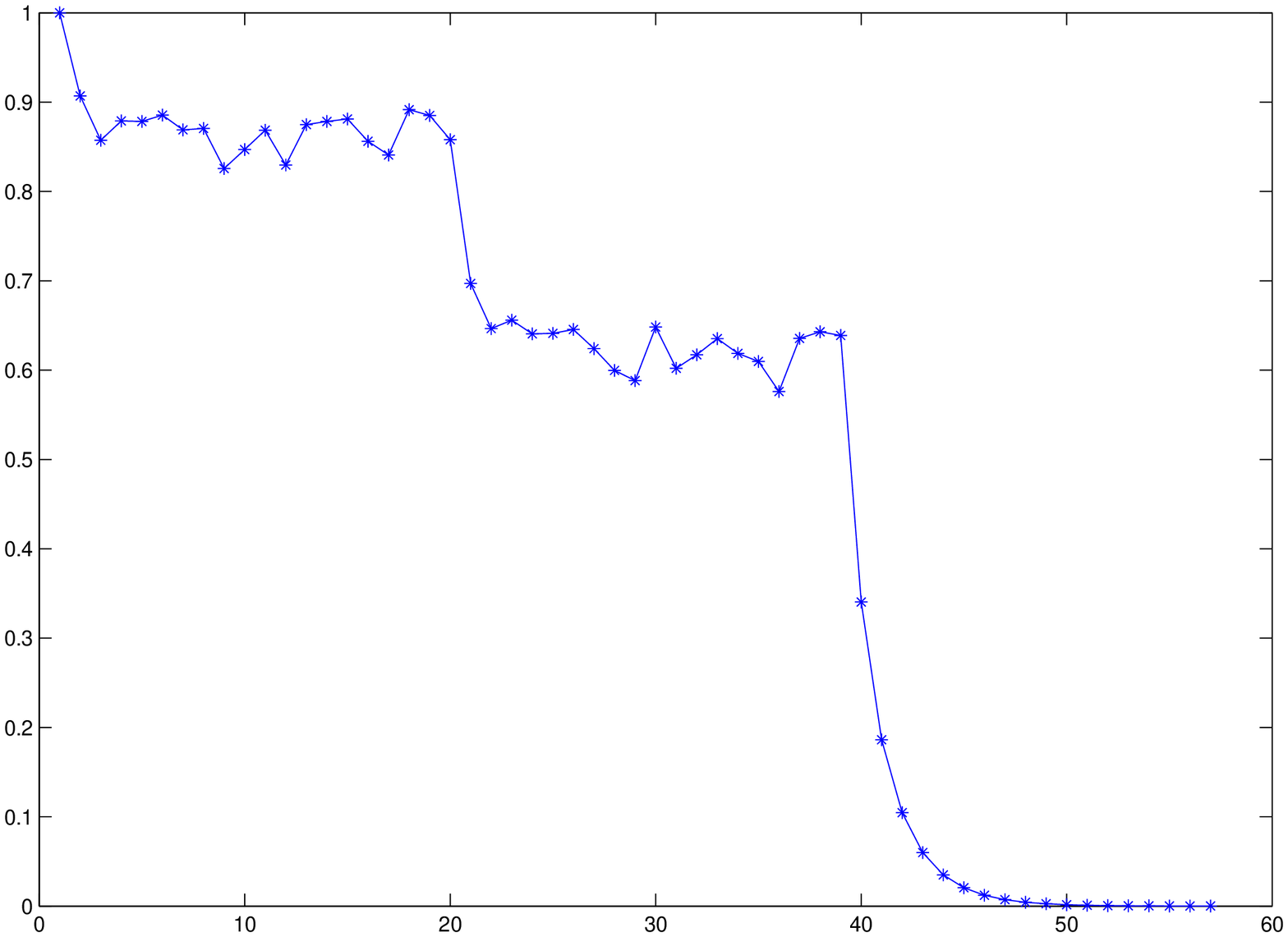}
      \label{converge}
      }
      \caption{A simulation result.}\label{sim1_1}
\end{figure}


\section{Conclusion}\label{sec6}

In this paper, we develop a distributed algorithm to compute the locations of sensor nodes based on the barycentric coordinates. Two critical problems are solved, leading to the success of globally convergent localization. First, for a general configuration that does not require every node to be inside a convex hull of its neighbors, a range information based algorithm is proposed to determine the signs of the barycentric coordinates and therefore the barycentric coordinates. Second, a distributed iterative algorithm is obtained with the global convergence ensured diagonal pre-conditioner designed based on a partially distributed cluster scheme. Future work includes analysis of localizability and convergence rate of the proposed approach, and localization performance in the presence of measurement noises.

\end{document}